\documentclass[11pt, reqno, letter]{amsart}
\usepackage{setspace}
\usepackage[foot]{amsaddr}

\usepackage{pkgfile}

\newtheorem{assumption}{Assumption}
\newtheorem{theorem}{Theorem}
\newtheorem{proposition}{Proposition}

\allowdisplaybreaks

\title[Sensitivity Analysis when outcome is censored by death in IV models]{Sensitivity analyses for average treatment effects when outcome is censored by death in instrumental variable models}
\author{Kwonsang Lee$^{1}$, Scott A. Lorch$^{2}$, and Dylan S. Small$^{3}$}
\address{$^{1}$ Department of Biostatistics, Harvard School of Public Health, {\tt kwonsanglee.stat@gmail.com}}
\address{$^{2}$ Center for Outcomes Research, Children's Hospital of Philadelphia}
\address{$^{3}$ Department of Statistics, University of Pennsylvania}

\begin{document}

\begin{abstract}
Two problems that arise in making causal inferences for non-mortality outcomes such as bronchopulmonary dysplasia (BPD) are unmeasured confounding and censoring by death, i.e., the outcome is only observed when subjects survive. In randomized experiments with noncompliance, instrumental variable methods can be used to control for the unmeasured confounding without censoring by death.  But when there is censoring by death, the average causal treatment effect cannot be identified under usual assumptions, but can be studied for a specific subpopulation by using sensitivity analysis with additional assumptions. However, in observational studies, evaluation of the local average treatment effect (LATE) in censoring by death problems with unmeasured confounding is not well studied. We develop a novel sensitivity analysis method based on instrumental variable models for studying the LATE. Specifically, we present the identification results under an additional assumption, and propose a three-step procedure for the LATE estimation. Also, we propose an improved two-step procedure by simultaneously estimating the instrument propensity score (i.e., the probability of instrument given covariates) and the parameters induced by the assumption. We have shown with simulation studies that the two-step procedure can be more robust and efficient than the three-step procedure. Finally, we apply our sensitivity analysis methods to a study of the effect of delivery at high-level neonatal intensive care units on the risk of BPD. 
\end{abstract}

\keywords{causal inference, covariate balancing propensity score, neonatal intensive care, non-mortality outcome, observational study, perinatal regionalization}



\maketitle

\doublespace

\section{Introduction}\label{s:intro}

\subsection{Does delivery at a high-level neonatal intensive care unit (NICU) really increase the risk of bronchopulmonary dysplasia (BPD)? }

Regionalization of health care provides high-quality and specialized care to manage a given type of illness for a targeted population. Within pediatric medicine, premature newborn infants or high-risk mothers are directed to facilities such as neonatal intensive care units (NICUs) with technical expertise and appropriate capabilities. Regionalized perinatal systems were developed in 1970s. During the subsequent decades, NICUs began to save very low birth weight (VLBW) (i.e., birth weight < 1500g) infants, and neonatal mortality rates consistently decreased as the number of high-risk mothers transferred to high-level NICUs increased. However, along with changes in the economics of health care, regionalized perinatal care began to weaken in many areas of the United States by the 1990s \cite{howell2002} and NICU services began to diffuse from regional centers to community hospitals. 

While many studies have shown that delivery at high-level NICUs reduces mortality rates \cite{phibbs2007, baiocchi2010, lorch2012, guo2014}, other studies that have examined the non-mortality outcome of bronchopulmonary dysplasia (BPD) have provided conflicting evidence about whether high level NICUs are effective in preventing it \cite{warner2004, lapcharoensap2015}. BPD is a chronic lung disease that results from damage to a premature infant' lungs by a breathing machine or long-term use of oxygen. A naive analysis of the Pennsylvania NICU data shows a positive association between delivery at high-level NICUs and BPD; the prevalence of BPD is 2.48\% (high-level NICUs) and 0.68\% (low-level NICUs) in the premature infant population. This comparison implies a substantial increase in the risk of BPD when delivered at high-level NICUs. However, it does not account for mortality rates simultaneously, which could lead to a biased analysis of the effect of delivery at high-level NICUs. For instance, there may exist some premature infants who could survive only if they are delivered at high-level NICUs, but could not avoid other complications. In this case, it is possible that a lower risk of later complications in low-level NICUs might be estimated while estimating a higher mortality rate. When considering mortality-morbidity composite outcomes, a recent study showed that delivery at high-levels indeed lowers the risk of death or other complications \cite{jensen2015}, contrary to the naive analysis.  

In this paper, we focus on estimating the effect of delivery at high-level NICUs on the risk of BPD by accounting for death. More precisely, without considering mortality-morbidity composite outcomes, we seek to answer the question ``Does delivery at high-level NICUs lower the risk of BPD as it lowers the risk of death?'' The question has not been well studied. This is mainly because mortality rates are considered as the primary outcome of interest, but also because other complications of premature infants are observed only after they survived. The latter problem is often called ``censoring by death'' problems. As discussed by many authors \cite{robins1995, frangakis2002}, a meaningful causal effect of the treatment is defined in the subset of premature infants that would survive regardless of whether they were delivered at high-level NICUs or low-level NICUs. This subset selection has been referred to as principal stratification \cite{frangakis2002}, and this subset is called the \textit{always-survived} principal stratum. Methods have been developed for estimating causal treatment effects in this stratum under certain model assumptions. Some of the methods propose sensitivity analysis methods by examining model assumptions with various values of the model parameters \cite{gilbert2003, shepherd2006, shepherd2007, jemiai2007}. Other methods provide bounds on the treatment effects with or without additional assumptions \cite{imai2008, yang2016}.

For estimating the treatment effect in the \textit{always-survived} stratum, most of the discussed methods have been developed in randomized experiments. However, in observational studies such as the NICU example, medical decisions are up to mothers, possibly with suggestions from doctors. Therefore, treatment cannot be assigned at random by researchers. A naive comparison between high-level and low-level NICUs can be misleading because high-risk mothers tend to be admitted to high-level NICUs and high-risk premature infants tend to be transferred to high-level NICUs. To account for such confounders, instrumental variable (IV) models can be considered by allowing for consistent estimation of the treatment effect. When there is a proper variable that can serve as an IV (we will discuss the IV assumptions in more details in Section~\ref{ss:notation}), the treatment effect can be consistently estimated for a subpopulation (the \textit{compliers}) in an observational study. The \textit{compliers} are the subpopulation that would take a treatment only if the treatment is encouraged. For the combined stratum of the \textit{always-survived} stratum and the \textit{compliers} subpopulation, we develop a sensitivity analysis method for the causal effect of delivery at high-level NICUs on the risk of BPD. 

\subsection{Data: travel time, NICU level, survival, BPD and covariates}
\label{ss:data}

The data consists of information on premature infants in Pennsylvania between 1995-2005 obtained from combining birth certificates, death certificates and hospital records \cite{lorch2012}. We consider two study populations from the data: (1) premature infants with a birth weight between 400-8000g and a gestational age between 23-37 weeks and (2) infants with a birth weight between 500-1500g and a gestational age between 23-37 weeks. For brevity, we call the first study population the \textit{premature infant} population, and call the second study population the \textit{VLBW infant} population. The VLBW infant population is the subset of the premature infant population. The first population can provide an overall inference about the treatment effect for general premature infants, and the second population is more targeted to the subset of VLBW infants who are already at a high risk of death or other complications. There are 174,878 premature infants in the data for the premature infants population, and there are 13,658 infants for the VLBW infants population. 

To compare the effectiveness of neonatal care provided by different levels of NICUs, we define a binary treatment variable with two levels, high-level and low-level NICUs. Based on previous studies \cite{phibbs1996, phibbs2007}, a high-level NICU is defined as a level 3 NICU with high-volume, which means that they have the capacity for sustained mechanical assisted ventilation and a minimum of 50 premature infants per year. We define a NICU as a low-level NICU if either it was below level 3 or it delivered an average of fewer than 50 premature infants per year. Also, we define a survival indicating variable as the opposite of in-hospital mortality. The in-hospital mortality is measured with two metrics, deaths during the initial hospitalization (neonatal deaths) and fetal deaths with a gestational age$\geq$ 23 weeks and a birth weight $\geq$ 400g. In addition, we consider a binary outcome of BPD that was obtained from the ICD-9-CM codes. In case when an infant did not survive, we assume that the BPD outcome of the infant is undefined. 

To make causal inference from the data, we consider the \textit{excess travel time} as an IV. The excess travel time is the time difference between the time traveling to the nearest high-level NICU and the time traveling to the nearest low-level NICU from a mother's residence. This time variable is calculated by using ARcView software (ESRI). Since a high-risk mother is more likely to deliver at a high-level NICU when it is close to home, we use this proximity to high-level NICUs compared to low-level NICUs. Specifically, we use a binary variable that indicates whether the excess travel time is less than or equal to five minutes. The validity of the excess travel time as an IV is discussed in previous studies \cite{baiocchi2010, guo2014}. Also, to control for confounders, we consider the following set of covariates: gestational age, birth weight, prenatal care, health insurance, single birth, parity, mother's age, race, education, and mother's complications.

\section{Notation, Causal Estimand, and Identification}
\label{s:framework}

\subsection{Notation}
\label{ss:notation}

Consider a study with $N$ premature infants. Each infant has an independent and identically distributed sample $(\mathbf{X}_i, Z_i, D_i, S_i, Y_i)$ from a given population of interest with $\mathbf{X}_i$ denoting a set of covariates, $Z_i$ a binary instrument, $D_i$ a binary treatment, $S_i$ a survival indicator, and $Y_i$ an outcome of interest. Let $Z_i=1$ if subject $i$ received encouragement to take a treatment, and $Z_i=0$ if subject $i$ did not receive encouragement. The encouragement may not be randomized, and may depend on covariates. Let $D_i=1$ if subject $i$ actually received a treatment and $D_i=0$ if subject $i$ did not. The survival indicator $S_i$ is an intermediate variable after receiving treatment (often, called a \textit{post-treatment} variable) denoting $S_i=1$ if infant $i$ survived and $S_i=0$ if did not survive. The outcome $Y_i$ is measured only when infant $i$ survived (i.e., $S_i=1$) otherwise $Y_i$ is undefined or missing.

To define our causal estimand, we adopt the potential outcome framework proposed by Neyman \cite{neyman1923} and Rubin \cite{rubin1974}. Define $D_i(1)$ to be the treatment indicator if $Z_i$ were to be set to 1 and $D_i(0)$ to be the value if $Z_i$ were to be set to 0. Similarly, $S_i(z, d)$ to be the survival indicator if $Z_i$ were to be set to $z$ and $D_i$ were to be set to $d$. Also, $Y_i(z, d, s)$ to be the outcome if $(Z_i, D_i, S_i)$ were to be set to $(z, d, s)$.  Since the outcome $Y_i$ is undefined when infant $i$ could not survive, we define $Y_i(z, d, 0) = * $. As discussed in the IV literature, the following assumption is considered. 

\begin{assumption}
	The IV assumptions are as follows:
	\begin{itemize}
		\item[(i)] Stable Unit Treatment Value Assumption (SUTVA) 
		
		\item[(ii)] Exclusion restriction: $S_i(0, d) = S_i(1, d)$ for $d=0$ or $1$ and $Y_i(0,d, s) = Y_i(1, d, s)$ for $\forall (d, s)$. 
		
		\item[(iii)] Monotonicity: $D_i(1) \geq D_i(0)$ and $S_i(z, 1) \geq S_i(z, 0)$ for $\forall z$. 
		
		\item[(iv)] Instrumentation:$P(S(z,1)=S(z,0)=1, D(0) < D(1)) \geq \epsilon > 0, \forall z$ .

		\item[(v)] The instrumental variable $Z_i$ is independent of the potential outcomes $Y_i(z,d,s)$, the potential survival indicator $S_i(z, d)$, and the potential treatment $D_i(z)$ conditioning on the covariates $\mathbf{X}$. 
		\[
		Z_i \perp\!\!\!\perp \Big( Y_i(z, d, s), S_i(z, d), D_i(z)  \Big) | \mathbf{X}
		\]
		
		\item[(vi)] Positivity: $0 < P(Z=1 | \mathbf{X}) < 1$. 
		
	\end{itemize} 
	\label{assumption1}
\end{assumption}

\begin{table}
\caption{\label{tab:description}Description of subpopulation. Upper left: the first compliance class $C_1$,  Upper right: the second compliance class $C_2$, and Lower: the combined compliance class $(C_1, C_2)$ for survived subjects.}
	\parbox{.45\linewidth}{
		\centering
		\begin{tabular}{l c | c c}
			\hline
			& & Low & High \\
			& & D=0 & D=1 \\
			\hline
			Far & Z=0 & NT/CO & AT \\
			Close & Z=1 & NT & AT/CO \\
			\hline
		\end{tabular}
		
	}
	\hfill
	\parbox{.45\linewidth}{
		\centering
		\begin{tabular}{l c | c c}
			\hline
			& & Death & Survive \\
			& & S=0 & S=1 \\
			\hline
			Low & D=0 & NS/PR & AS \\
			High & D=1 & NS & AS/PR\\
			\hline
		\end{tabular}
	}
	\begin{center}
		\begin{tabular}{c|cc|cc}
		\hline
		& \multicolumn{2}{c|}{$S=1, D=0$} & \multicolumn{2}{c}{$S=1, D=1$} \\
		\hline
		$Z=0$ & NT & CO & AT & \\
		& (AS) & (AS) & (AS/PR) & \\
		\hline
		$Z=1$ & NT & & AT & CO \\
		& (AS)& & (AS/PR) & (AS/PR) \\
		\hline
	\end{tabular}
	\end{center}
\end{table}

Assumption~\ref{assumption1} (i) is that the outcome (treatment) for infant $i$ is not affected by the values of the treatment or instrument (instrument) for other infants and that the outcome (treatment) does not depend on the way the treatment or instrument (instrument) is administered. This SUTVA assumption allows us to use the notation $Y_i(z, d, s)$ (or $S_i(z, d), D_i(z)$) which means that the outcome (treatment) for infant $i$ is not affected by the values of the treatment and instrument (instrument) for other infants. See Angrist, Imbens, and Rubin \cite{angrist1996} for more discussions of the SUTVA assumption. Assumption~\ref{assumption1} (ii) assures that any effect of $Z$ on $S$ must be through an effect of $Z$ on $D$. Under this assumption, we can write the potential survival indicator as $S_i(d)$ instead of $S_i(z,d)$, and we can reduce the form of $Y_i(z, d, s)$ as $Y_i(d, s)$. Regarding Assumption~\ref{assumption1} (iii), there are three possible compliance classes $C_1$ depending on $D_i(0)$ and $D_i(1)$: \textit{Always-takers (AT)} if $D_i(0)=D_i(1)=1$, \textit{Never-takers (NT)} if $D_i(0)=D_i(1)=0$, and \textit{Compliers (CO)} if $D_i(0) < D_i(1)$. The upper left table of Table~\ref{tab:description} shows what the population consists of in terms of three classes of $C_1$. Similarly, there are three possible compliance classes $C_2$ depending on $S_i(0)$ and $S_i(1)$: \textit{Always-survived (AS)} if $S_i(0)=S_i(1)=1$, \textit{Never-survived (NS)} if $S_i(0)=S_i(1)=0$, and \textit{Protected (PR)} if $S_i(0) < S_i(1)$. The upper right table of Table~\ref{tab:description} shows what the population consists of in terms of three classes of $C_2$. We define $C$ as a composite compliance class, which is $C = C_1 \times C_2$. For simplicity, for example, we denote $C=$\textit{CO-AS} if $C_1=$\textit{CO} and $C_2=$\textit{AS}. Assumption~\ref{assumption1} (iv) imposes a positive value of the proportion of the \textit{compliers} \& \textit{always-survived} stratum, called the \textit{CO-AS} stratum. Assumption~\ref{assumption1} (v) will be satisfied if the confounders between $Z_i,  S_i, Y_i$ are controlled by the covariate $\mathbf{X}_i$. Assumption~\ref{assumption1} (vi) implies that every infant has positive values of receiving encouragement.

\subsection{Causal Estimand}

We can define the causal treatment effect on $Y$ based on two potential outcome $Y(1,1)$ and $Y(0,1)$. The difference $Y(1,1)-Y(0,1)$ is measured and considered as our treatment effect. However, for a binary outcome, other measures such as risk ratio or odds ratio can be considered. When the outcome is censored by death, the potential outcomes $Y(1, 0)$ and $Y(0,0)$ are undefined as discussed in Section 2.1. We argue that it may not be meaningful to compare $Y(1, 0)$ and $Y(0,0)$ since it is an outcome comparison when an infant would not survive. Therefore, we restrict our attention to those who would always survive regardless of treatment status. For infants in the \textit{AS} stratum shown in the upper right table of Table~\ref{tab:description},  the risk difference $Y(1,1)-Y(0,1)$ can be used to quantify the treatment effect. Other risk measures such as the risk ratio or the odds ratio can be considered. 

To identify the causal treatment effect, we need to consider a more refined stratum than the \textit{AS} stratum. Since treatment is not randomized and only encouragement is administered to stimulate subjects to take a treatment, our attention is further restricted to those who respond to encouragement. As in the usual IV analysis, the treatment effect can only be identified for this subgroup, the \textit{CO} stratum. To combine censoring by death with the IV approach, we can focus on the composite compliance class, the \textit{CO-AS} stratum. Our goal in this paper is making inference for the local average treatment effect (LATE) for the \textit{CO-AS} stratum. Using the risk difference measure, the LATE is defined as
\[
LATE = E(Y(1,1)-Y(0,1) | S(0)=S(1)=1, D(0)<D(1)) = E(Y(1,1)-Y(0,1)|C=\textit{CO-AS}).
\]
For cases when censoring by death does not occur, the average treatment effect for the \textit{CO} stratum can be identified with the IV analysis. As shown in the upper left table of Table~\ref{tab:description}, the usual IV analysis separates the information for compliers without treatment from the cell of $Z=0, D=0$ by taking out the information for never-takers using the cell of $Z=1, D=0$. Similarly, the information for compliers with treatment can be obtained from the cells of $Z=1, D=1$. However, when outcome is censored, this analysis cannot be applied, thus needs modification. The lower table of Table~\ref{tab:description} shows the five composite compliance classes $C$. From this table, using the same technique as the IV analysis, the information for the \textit{CO-AS} stratum without treatment can be obtained, but the information for the \textit{CO-AS} stratum with treatment cannot. Only the mixed information for the \textit{CO-AS} (with treatment) and \textit{CO-PR} strata can be obtained. The following assumption enables the information for the \textit{CO-AS} stratum with treatment to be separated, thus identified. 

\begin{assumption}
	Assume a the model for mixing probabilities of the \textit{always-survived (AS)} and \textit{protected (PR)} in the treated compliers, 
	\begin{equation}
	\Pr(S(0)=1|S(1)=1, D(0)<D(1), Y(1, 1)) = w(Y(1,1); \beta)
	\label{model1}
	\end{equation}
	where $w(y; \beta) = \Phi(\alpha + g(y; \beta ))$, $\beta$ is fixed and known, $\Phi(\cdot)$ and $g(\cdot)$ are  known function, but $\alpha$ is not specified.  
	
	\label{assumption2}
\end{assumption}
\noindent The probability in model~\eqref{model1} can be represented as $\Pr(C=\textit{CO-AS} \>|\> C = \textit{CO-AS} \>\>\text{or}\>\> \textit{CO-PR}, Y(1,1))$. The mixing probabilities depend on the levels of the potential outcome $Y(1,1)$, which restricts the relationship of outcome between the \textit{CO-AS} and \textit{CO-PR} strata. For instance, if the potential outcome distributions for the $C=\textit{CO-AS}$ and $C=\textit{CO-PR}$ are identical, the function $w$ does not depend on $Y(1,1)$.

 The parameter $\beta$ cannot be identified from the observed data. We propose regarding $\beta$ as fixed and known, and then investigating the impact of $\beta$ on estimation of the LATE by considering a plausible range of $\beta$. Shepherd et al.\cite{shepherd2006} considered a similar assumption in the setting when a treatment is randomized while including covariates $X$. In our model, we use marginal mixing probabilities. 

\subsection{Identification}
\label{ss:identification}

In this subsection, we describe identification results under Assumption~\ref{assumption1} and \ref{assumption2}. 

\begin{theorem} \normalfont
	Under Assumption~\ref{assumption1} and \ref{assumption2}, 
	\begin{itemize}
		\item[(i)] The average potential outcome $Y(0,1)$ for the \textit{CO-AS} stratum is identified, 
		$$
		E(Y(0,1) | S(0)=S(1)=1, D(0)<D(1)) = \frac{E\left[\frac{YS(1-D)(1-Z)}{1-e(X)}\right] - E\left[ \frac{YS(1-D)Z}{e(X)}\right]}{\Pr(S(0)=S(1)=1, D(0)<D(1))}
		$$
		where $e(X)=\Pr(Z=1|X)$. 
		
		\item[(ii)] For a known value of $\beta$, the average potential outcome $Y(1,1)$ for the \textit{CO-AS} stratum is identified, 
		$$
		E(Y(1,1) | S(0)=S(1)=1, D(0)<D(1)) = \frac{E\left[ \frac{Y w(Y; \beta) SDZ}{e(X)}\right] - E\left[ \frac{Y w(Y; \beta) SD(1-Z)}{1-e(X)} \right]}{\Pr(S(0)=S(1)=1, D(0)<D(1))}.
		$$
		
		\item[(iii)] The proportion of the \textit{CO-AS} stratum can be identified, 
		$$
		\Pr(S(0)=1, D(0)<D(1)) = {E}\left[\frac{S(1-D)(1-Z)}{1-{e}(X)}\right] - {E}\left[
		\frac{S(1-D)Z}{{e}(X)}\right].
		$$

	\end{itemize}
	
	\label{thm1}
\end{theorem}

\begin{proof}
	See Appendix~\ref{ss:proof_thm1}.
\end{proof}

From Theorem~\ref{thm1}, the average causal treatment effect for the \textit{CO-AS} stratum, $\{S(0)=S(1)=1, D(0)<D(1)\}$, can be identified as
\begin{align}
& E[Y(1,1) -Y(0,1) | S(0)=S(1)=1, D(0)<D(1)] \nonumber \\
&= \frac{\left( E\left[ \frac{Y w(Y; \beta) SDZ}{e(X)}\right] - E\left[ \frac{Y w(Y; \beta) SD(1-Z)}{1-e(X)} \right]\right) -\left( E\left[\frac{YS(1-D)(1-Z)}{1-e(X)}\right] - E\left[ \frac{YS(1-D)Z}{e(X)}\right] \right)}{E\left[\frac{S(1-D)(1-Z)}{1-{e}(X)}\right] - {E}\left[	\frac{S(1-D)Z}{{e}(X)}\right]}.
\label{eqn:late}
\end{align}
Equation~\eqref{eqn:late}  is similar to the identification results in Abadie \cite{abadie2003} and Tan \cite{tan2006}, but has a more complicated form since these authors did not consider censoring by death problems.

For the identification of the LATE, Assumption~\ref{assumption2} determined the mixing probabilities between the $C_1=CO, C_2=AS$ and $C_1=CO, C_2=PR$ subpopulations. Furthermore, it inherently implies one more relationship. The following proposition shows this relationship. 

\begin{proposition}
	Given the function $w(y; \beta)$, 
	\begin{align}
	E\left[\frac{w(Y; \beta) SDZ}{e(X)} \right] - E\left[\frac{w(Y; \beta) SD(1-Z)}{1-e(X)} \right] &= \Pr(S(0)=S(1)=1, D(0)<D(1)) \nonumber \\
	&= E\left[\frac{S(1-D)(1-Z)}{1-{e}(X)}\right] - {E}\left[	\frac{S(1-D)Z}{{e}(X)}\right].
	\label{eqn:prop}
	\end{align}
	\label{prop1}
\end{proposition}
\begin{proof}
	
	See Appendix~\ref{ss:proof_prop1}.
	
\end{proof}
\noindent To understand Proposition~\ref{prop1}, we can divide both sides of \eqref{eqn:prop} by $\Pr(S(1)=1, D(0)<D(1))$, which produdes a different representation as $E[w(Y(1,1); \beta) | S(0)=S(1)=1, D(0)<D(1)] = \Pr(S(0)=1 | S(1)=1, D(0) <D(1))$. This representation implies that the expected value of the mixing probability $w(Y(1,1); \beta)$ over $Y(1,1)$ must be equal to the marginal mixing probability $\Pr(S(0)=1 | S(1)=1, D(0) <D(1))$. It is an implicit assumption by considering Assumption~\ref{assumption2}. This proposition will be important in the next section for estimating the unspecified parameter $\alpha$ in Assumption~\ref{assumption2}.  

\section{Estimation and Inference}
\label{s:method}

\subsection{Estimation of the LATE}
\label{ss:estimation}

In this section, we illustrate our method for estimating the LATE from the observed data. Tan \cite{tan2006} and Cheng and Lin \cite{cheng2017} discussed a weighting method for the estimation of the LATE in the usual IV setting in the presence of covariates by introducing the instrument propensity score $e(\mathbf{X}) = \Pr(Z=1| \mathbf{X})$. From Assumption~\ref{assumption1}, $e(\mathbf{X})$ is assumed to be strictly positive. We will work with this conditional probability for estimation. However, $e(\mathbf{X})$ is generally unknown, so it has to be estimated from the observed data. A simple approach for estimating $e(\mathbf{X})$ is using logistic regressions. For instance, we can assume a model, $e(\mathbf{X}; \gamma) = \exp( \gamma^T \mathbf{X}) / (1+ \exp( \gamma^T \mathbf{X}))$. Then, the parameter $\gamma$ can be estimated, for instance, through maximizing log-likelihood function: $\hat{\gamma}=\arg \max_{\gamma} \sum_{i=1}^{N} Z_i \log\{ e(\mathbf{X}_i) \} + (1-Z_i) \log \{ 1- e(\mathbf{X}_i)\}$. We denote the estimate by $\hat{e}(\mathbf{X}) = e(\mathbf{X}; \hat{\gamma})$. 

In addition to $e(\mathbf{X})$, the specific model of the function $w(\cdot; \beta)$ needs to be specified. For later simulation and application, we assume that the function $w(\cdot; \beta)$ has an \textit{expit} function, 
\begin{equation}
w(; \beta) = w(y; \alpha, \beta) = \frac{\exp(\alpha+ \beta y)}{1+\exp(\alpha+\beta y)} \equiv \text{expit}(\alpha + \beta y). 
\label{eqn:w}
\end{equation} 
The function $w(y; \alpha, \beta)$ depends on two parameters $\alpha$ and $\beta$, but only $\alpha$ can be estimated from the observed data while $\beta$ is fixed and known. Since, in practice, $\beta$ is not known, we conduct a sensitivity analysis of $\beta$ by examining several values of $\beta$ with a plausible range of $\beta$. We will present the sensitivity analysis in Section~\ref{s:simulation} with simulations, and also in Section~\ref{s:example} using the NICU hospital data.

From Proposition~\ref{prop1}, the parameter $\alpha$ is implicitly determined given that $\beta$ is known. By plugging the estimate $\hat{e}(\mathbf{X}_i)$ into \eqref{eqn:prop}, the proportion $\Pr(S(0)=S(1)=1, D(0)<D(1))$ can be nonparametrically estimated in two different directions. First, the proportion can be estimated by subtracting the proportion of $C_1=AS, C_2=NT$ from the proportion of $C_1=AS, C_2=NT$ or $CO$,
\begin{equation}
\hat{\Pr}(S(0)=S(1)=1, D(0)<D(1)) = \frac{1}{N}\sum_{i=1}^{N} \left\{ \frac{S_i(1-D_i)(1-Z_i)}{1-\hat{e}(\mathbf{X}_i)} - \frac{S_i(1-D_i)Z_i}{\hat{e}(\mathbf{X}_i)} \right\}.
\label{eqn:first_way}
\end{equation}
Alternatively, this proportion can be estimated as
\begin{equation}
\hat{\Pr}(S(0)=S(1)=1, D(0)<D(1)) = \frac{1}{N}\sum_{i=1}^{N} \left\{ \frac{w(Y_i;\alpha, \beta)S_i D_i Z_i}{\hat{e}(\mathbf{X}_i)} - \frac{w(Y_i;\alpha, \beta)S_i D_i (1-Z_i)}{1-\hat{e}(\mathbf{X}_i)} \right\}.
\label{eqn:second_way}
\end{equation}
Therefore, the estimate $\hat{\alpha}$ can be found by equating the two estimation equations. More formally, this can be viewed as solving the equation $h(\alpha)=0$ where
\begin{align*}
h(\alpha) &= \frac{1}{N}\sum_{i=1}^{N} \left\{ \frac{w(Y_i;\alpha, \beta)S_i D_i Z_i}{\hat{e}(\mathbf{X}_i)} - \frac{w(Y_i;\alpha, \beta)S_i D_i (1-Z_i)}{1-\hat{e}(\mathbf{X}_i)} \right\} - \left\{ \frac{S_i(1-D_i)(1-Z_i)}{1-\hat{e}(\mathbf{X}_i)} - \frac{S_i(1-D_i)Z_i}{\hat{e}(\mathbf{X}_i)} \right\} \\
&= \frac{1}{N}\sum_{i=1}^{N} \left\{ \frac{S_i (w(Y_i; \alpha, \beta)D_i + 1- D_i)Z_i}{\hat{e}(\mathbf{X}_i)} - \frac{S_i(w(Y_i; \alpha, \beta)D_i + 1-D_i)(1-Z_i)}{1-\hat{e}(\mathbf{X}_i)} \right\}
\end{align*}
The estimate $\hat{\alpha}$ then be plugged in the function $w(y; \alpha, \beta)$, and we call it $\hat{w}(y) = w(y; \hat{\alpha}, \beta)$. 

 Next, for the last step of the estimation, we replace $e(\mathbf{X})$ and $w(y;\alpha, \beta)$ by $\hat{e}(\mathbf{X})$ and $\hat{w}(y)$ in equation~\eqref{eqn:late}. Then, the LATE is nonparametrically estimated as 
\begin{align}
	\widehat{\text{LATE}}(\beta) = \frac{\frac{1}{N} \sum_{i=1}^{N}\left\{ \frac{Y_i S_i(\hat{w}(Y_i)D_i + 1-D_i)Z_i}{\hat{e}(\mathbf{X}_i)} -  \frac{Y_i S_i(\hat{w}(Y_i)D_i + 1-D_i)(1-Z_i)}{1-\hat{e}(\mathbf{X}_i)} \right\}}{ \frac{1}{N}\sum_{i=1}^{N} \left\{ \frac{S_i(1-D_i)(1-Z_i)}{1-\hat{e}(\mathbf{X}_i)} - \frac{S_i(1-D_i)Z_i}{\hat{e}(\mathbf{X}_i)} \right\} }.
	\label{eqn:late_estimate}
\end{align}

To summarize, estimation of the LATE can be viewed as a three-step procedure,
\begin{enumerate}
	\item[Step 1.] Using the instrument $Z$ and covariates $\mathbf{X}$, estimate the instrument propensity score $e(\mathbf{X})$, which is denoted as $\hat{e}(\mathbf{X})$, 
	
	\item[Step 2.] For a fixed $\beta$, estimate $\alpha$ by plugging in $\hat{e}(\mathbf{X})$ and solving $h(\alpha)=0$, 
	
	\item[Step 3.] Plug-in $\hat{\alpha}$ and $\hat{e}(X)$ into equation~\eqref{eqn:late_estimate}, and then compute an estimate $\widehat{\text{LATE}}(\beta)$. Repeat this procedure for other values of $\beta$.
	
\end{enumerate}

\noindent During the procedure, the instrument propensity score $e(\mathbf{X}_i)$ and $\alpha$ are sequentially estimated. However, in practice, this two step estimation may not work because the root of the function $h(\alpha)$ in Step 2 may not exist sometimes. Actually, we found through simulation studies that this problem occasionally arises. For instance, since $0 \leq w(;\alpha, \beta) \leq 1$, equation~\eqref{eqn:second_way} is smaller than or equal to $\frac{1}{N}\sum_{i=1}^{N} \left\{ \frac{S_i D_i Z_i}{\hat{e}(\mathbf{X}_i)} - \frac{S_i D_i (1-Z_i)}{1-\hat{e}(\mathbf{X}_i)} \right\}$ that is the estimator of the proportion $\widehat{\Pr} (S(1)=1, D(0)<D(1))$. The estimate $\widehat{\Pr} (S(1)=1, D(0)<D(1))$ must be larger than the estimate $\widehat{\Pr}(S(0)=S(1)=1, D(0)<D(1))$, however it is often smaller due to the estimation variability of the propensity score $e(\mathbf{X}_i)$. In this case, for every value of $\alpha$ in $(-\infty, \infty)$, equation~\eqref{eqn:second_way} is strictly smaller than equation~\eqref{eqn:first_way}, resulting in no solution to $h(\alpha)=0$. One way to avoid this problem is to estimate $e(\mathbf{X})$ and $\alpha$ within the same framework. We propose an approach for this simultaneous estimation in the next subsection.

\subsection{Parameter estimation based on the covariate balancing propensity score}
\label{ss:cov_balancing}

To estimate the instrument propensity score $e(\mathbf{X})$ and $\alpha$ simultaneously, we can extend the {\it covariate balancing propensity score} (CBPS) estimation proposed by Imai and Ratkivic (2014). The CBPS estimation method is designed to estimate $e(\mathbf{X})$ while optimizing covariate balance. Inverse propensity score weighting asymptotically achieves covariate balance, but CBPS tries to improve balance in finite sample. The CBPS estimation manually optimizes sample covariate balance between treated and control groups. When the model for the propensity score is not correctly specified, the covariate balancing property may not hold, so enforcing the covariate balance can increase the robustness of model misspecification. 

The CBPS estimation can be done within the framework of generalized method of moments (GMM) by minimizing covariate balance measure $\frac{1	}{N} \sum_{i=1}^{N} \left\{ \frac{Z_i \mathbf{X}_i}{e(\mathbf{X}_i)} - \frac{(1-Z_i)\mathbf{X}_i}{1-e(\mathbf{X}_i)} \right\}$. For the estimation of $\alpha$ we use $h(\alpha)$, and this function can be represented by $h(\alpha)= \frac{1	}{N} \sum_{i=1}^{N} \left\{ \frac{Z_i W_i}{e(\mathbf{X}_i)} - \frac{(1-Z_i)W_i}{1-e(\mathbf{X}_i)} \right\}$ by letting $W_i = S_i(w(Y_i; \alpha, \beta)D_i + 1-D_i)$. Since estimating $\alpha$ using $h(\alpha)$ is equivalent to minimizing the covariate balance measure for $W_i$, we can assume that $W_i$ is an additional covariate, and minimize the sample covariate balance measure for an extended covariate vector $\tilde{\mathbf{X}}_i = (\mathbf{X}_i, W_i)$. As we discussed in Section~\ref{ss:estimation}, the logistic model $e(\mathbf{X}; \gamma)$ is considered for the propensity score. We can have an extended estimator $\hat{\tilde{\gamma}} = (\hat{\gamma}, \hat{\alpha})$ by using the efficient GMM estimator proposed by Hansen\cite{hansen1982}. The GMM estimator $\hat{\tilde{\gamma}}_{GMM}$ is 
\begin{align*}
\hat{\tilde{\gamma}}_{GMM} &= \arg \min_{\tilde{\gamma}} g_{\tilde{\gamma}} (\mathbf{Z}, \mathbf{\tilde{X}})^T \Sigma_{\tilde{\gamma}} (\mathbf{Z}, \mathbf{\tilde{X}})^{-1} {g}_{\tilde{\gamma}} (\mathbf{Z}, \mathbf{\tilde{X}}) 
\end{align*} 
where ${g}_{\tilde{\gamma}} (\mathbf{Z}, \mathbf{\tilde{X}}) = \frac{1}{N} \sum_{i=1}^{N} \left\{ \frac{Z_i \mathbf{\tilde{X}}_i}{e(\mathbf{X}_i)} - \frac{(1-Z_i)\mathbf{\tilde{X}}_i}{1-e(\mathbf{X}_i)} \right\}$ and $\Sigma_{\tilde{\gamma}} (\mathbf{Z}, \mathbf{\tilde{X}}) =  \frac{1}{N} \sum_{i=1}^{N} \frac{\mathbf{\tilde{X}}_i \mathbf{\tilde{X}}_i^T}{e(\mathbf{X}_i)\{1-e(\mathbf{X}_i) \} }$. See Imai and Ratkovic\cite{imai2014} for further details on the GMM-based CBPS estimation. After acquiring the estimates $\hat{e}(\mathbf{X}_i)$ and $\hat{w}(Y_i)$, we can plug-in them into equation~\eqref{eqn:late} to obtain the LATE estimate.  

The purpose of our approach is to incorporate the restriction for $\alpha$ within the propensity score estimation in order to satisfy the equality~\eqref{eqn:prop} in Proposition~\ref{prop1}. This approach is particularly useful when there is no solution to $h(\alpha)=0$. When $h(\alpha)=0$ is solvable and $\beta$ is correctly specified, the simultaneous estimation of $e(\mathbf{X}; \gamma)$ and $\alpha$ itself does not improve the performance of separate estimation of them. However, when $\beta$ is not correctly specified, there is an additional advantage of applying our proposed method; the bias of the LATE estimator tends to be smaller. We will discuss this using simulation studies in Section~\ref{s:simulation}.

\section{Simulation}
\label{s:simulation}

In this section, we examine the small sample performance of our proposed estimator for various simulation settings. For each setting, we set the sample size $N=2000$. Also, we consider four binary covariates $\mathbf{X}_i=(X_{i1}, X_{i2}, X_{i3}, X_{i4})$ and assume that each combination has $N/16=125$ subjects. This covariate matrix is fixed throughout simulation studies. For the data generating process, the instrument $Z_i$ is generated based on the true propensity $e(\mathbf{X}_i)=\Pr(Z_i=1 | \mathbf{X}_i) = expit(0.5+0.2X_{i1}-0.2X_{i2})$, i.e. $Z_i \sim Binom(e(\mathbf{X}_i))$. Let $$
q_{C_1, C_2} = (q_{\textit{CO-AS}}, q_{\textit{CO-PR}}, q_{\textit{CO-NS}}, q_{\textit{AT-AS}}, q_{\textit{AT-PR}}, q_{\textit{AT-NS}}, q_{\textit{NT-AS}}, q_{\textit{NT-PR}}, q_{\textit{NT-NS}})$$
 be the vector of the proportions of the nine composite compliance classes $C=C_1 \times C_2$. Then, the class membership $C_i$ is generated using a multinomial distribution with $q_{C_1, C_2}$. Given $C_i$,  $D_i$ and $S_i$ are determined. For instance, when $C_i=$\textit{CO-AS} and $Z_i=0$, $D_i=0$ and $S_i=1$. We consider binary outcomes that are generated by binomial distributions. Specifically, we assume $Y_i|C_i={\textit{CO-AS}}, Z_i=0 \sim Binom(p_{{\textit{CO-AS}},0})$, $Y_i|C_i={\textit{CO-AS}}, Z_i=1 \sim Binom(p_{{\textit{CO-AS}},1})$ and $Y_i|C_i={\textit{CO-AS}}, Z_i=0 \sim Binom(p_{{\textit{CO-PR}},1})$. Therefore, the LATE is defined as $p_{{\textit{CO-AS}},1} - p_{{\textit{CO-AS}},0}$. We set $p_{{\textit{CO-AS}},0}=0.3$, and, for other compliance classes, we assume a binomial distribution $Binom(0.3)$. Given $\alpha$ and $\beta$, $p_{\textit{CO-AS}, 1}$ and $p_{{\textit{CO-PR}},1}$ are determined from Proportion~\ref{prop1}: 
\begin{align*}
p_{{\textit{CO-AS}},1} &= p_{\textit{CO-AS  or  PR}} \times expit(\alpha+\beta) \frac{q_{\textit{CO-AS}}+q_{\textit{CO-PR}}}{q_{\textit{CO-AS}}} \\
p_{{\textit{CO-PR}},1} &= p_{\textit{CO-AS  or  PR}} \times (1-expit(\alpha+\beta)) \frac{q_{\textit{CO-AS}}+q_{\textit{CO-PR}}}{q_{\textit{CO-PR}}}.
\end{align*}
where $p_{\textit{CO-AS  or  PR}} = \frac{q_{\textit{CO-AS}}/(q_{\textit{CO-AS}}+q_{\textit{CO-PR}}) - expit(\alpha)}{expit(\alpha+\beta) - expit(\alpha)}$ if $\beta \neq 0$ and $p_{\textit{CO-AS  or  PR}}$ can be an arbitrary probability if $\beta=0$.
 
In our simulation studies, we consider three simulation scenarios:
\begin{enumerate}
\item[(S1)] $\alpha=0, \beta=0$ with $q_{C_1, C_2} = (0.3, 0.3, 0.05, 0.1, 0.05, 0.05, 0.05, 0.05, 0.05)$. The true LATE is 0.2 with $p_{{\textit{CO-AS}},1}=0.5$.


\item[(S2)] $\alpha=0, \beta=3$ with $q_{C_1, C_2} = (0.4, 0.1, 0.05, 0.2, 0.05, 0.05, 0.05, 0.05, 0.05)$. The true LATE is 0.489 with $p_{{\textit{CO-AS}},1}=0.789$.

\item[(S3)] $\alpha=2, \beta=3$ with $q_{C_1, C_2} = (0.4, 0.01, 0.05, 0.29, 0.05, 0.05, 0.05, 0.05, 0.05)$. The true LATE is 0.558 with $p_{{\textit{CO-AS}},1}=0.858$.

\end{enumerate}
and also, we consider three methods:
\begin{enumerate}
\item[(a)] CBPS2 --- The simultaneous estimation of $e(\mathbf{X}_i)$ and $\alpha$ is considered as discussed in Section~\ref{ss:cov_balancing}, which emphasizes the two-step estimation procedure. 

\item[(b)] CBPS3 --- The three-step estimation procedure is considered for separate estimation. The propensity score is estimated by only considering covariates $\mathbf{X}_i$. The parameter $\alpha$ is then estimated by solving $h(\alpha)=0$ as discussed in Section~\ref{ss:estimation}.  

\item[(c)] GLM3 --- The standard logistic regression is used for estimating $e(\mathbf{X}_i)$ in the three-step estimation procedure. 
\end{enumerate}

\begin{table}
\centering
\caption{Sensitivity analysis for the LATE estimators from three approaches for estimating the instrument propensity score. }
\begin{tabular}{@{\extracolsep{4pt}}rlrr rrr c@{}}
\hline
& & \multicolumn{2}{c}{$\beta$} & \multicolumn{3}{c}{LATE} & Fail to solve \\
\cline{3-4}\cline{5-7}
& True ATE & \multicolumn{1}{c}{True} & \multicolumn{1}{c}{Assumed} & \multicolumn{1}{c}{CBPS2} & \multicolumn{1}{c}{CBPS3} & \multicolumn{1}{c}{GLM3} & $h(\alpha)=0$\\
\hline
S1 & 0.2 & 0 & -2 & -0.029 (0.044) & -0.223 (0.043) & -0.223 (0.043) & 0.000 \\
& & 0 & -1 & 0.078 (0.042) & -0.035 (0.043) & -0.035 (0.043) & 0.000 \\
& & 0 & 0 & 0.201 (0.042) & 0.201 (0.042) & 0.201 (0.042) & 0.000 \\
& & 0 & 1 & 0.323 (0.042) & 0.387 (0.045) & 0.387 (0.045) & 0.000 \\
& & 0 & 2 & 0.431 (0.044) & 0.489 (0.049) & 0.489 (0.049) & 0.000 \\[0.2cm]


S2 & 0.489 & 3 & 1 & 0.410 (0.043) & 0.332 (0.041) & 0.332 (0.041) & 0.000 \\
& & 3 & 2 & 0.455 (0.045) & 0.441 (0.041) & 0.441 (0.041) & 0.000 \\
& & 3 & 3 & 0.489 (0.049) & 0.489 (0.049) & 0.489 (0.049) & 0.000 \\
& & 3 & 4 & 0.510 (0.053) & 0.507 (0.053) & 0.507 (0.053) & 0.000 \\
& & 3 & 5 & 0.522 (0.056) & 0.515 (0.055) & 0.515 (0.055) & 0.000 \\[0.2cm]

S3 & 0.558 & 3 & 1 & 0.543 (0.052) & 0.510 (0.051) & 0.510 (0.051) & 0.295 \\
& & 3 & 2 & 0.552 (0.052) & 0.550 (0.048) & 0.550 (0.048) & 0.302 \\
& & 3 & 3 & 0.557 (0.054) & 0.564 (0.052) & 0.564 (0.052) & 0.298 \\
& & 3 & 4 & 0.562 (0.057) & 0.571 (0.055) & 0.571 (0.055) & 0.300 \\
& & 3 & 5 & 0.565 (0.058) & 0.575 (0.056) & 0.575 (0.056) & 0.296 \\
\hline
\end{tabular}
\label{tab:sim}
\end{table}

Table~\ref{tab:sim} shows the simulated sensitivity analysis results for the three simulation scenarios. From 10000 replications, the mean and standard deviation of 10000 LATE estimates are reported for each method. In the first and second scenarios, CBPS3 and GLM3 successfully solve the equation $h(\alpha)=0$ for every simulation as seen in the last column of Table~\ref{tab:sim}. When $\beta$ is correctly specified, all three methods produce unbiased estimators, but when $\beta$ is not the true value, all the estimators are biased. However, the bias of the CBPS2 estimator tends to be smaller than that of other estimators. This is because the CBPS2 method will divide the efforts invested for estimating $e(\mathbf{X}_i)$  and $\alpha$. It estimates $e(\mathbf{X}_i)$ while achieving equation~\eqref{eqn:prop} in Proposition~\ref{prop1}, but it does not directly attempt to solve $h(\alpha)=0$. This improves the overall performance of the CBPS2 estimator.

In the third scenario, the proportion of \textit{CO-AS} is $q_{\textit{CO-AS}}=0.4$ and the proportion of \textit{CO-PR} is $q_{\textit{CO-PR}}=0.01$. This small $q_{\textit{CO-PR}}$ implies that the estimate $\widehat{\Pr}(C = \textit{CO-AS})$ can be larger than the estimate $\widehat{\Pr}(C= \textit{CO-AS} \>\text{or}\> \textit{CO-PR})$ with some probability, which leads to a failure of solving $h(\alpha)=0$. When failing to solve $h(\alpha)=0$, $\alpha$ cannot be estimated in the three-step procedure; however, CBPS2 can be always applied to estimating $\alpha$. In addition, we found that CBPS3 cannot consistently estimate $\alpha$ even when $h(\alpha)=0$ is solvable. As a selected value of $\beta$ is far from the true value, $h(\alpha)$ is more difficult to solve, and thus  $\alpha$ can be over-estimated to compensate a low value of $\beta$ in our simulations (it can be under-estimated in a different simulation scenario). In Table~\ref{tab:sim}, the mean and standard deviation of the CBPS3 (and GLM3) estimator are reported only when $h(\alpha)=0$ is solvable. Even when $\beta$ is correctly specified, CBPS3 and GLM3 cannot provide an unbiased estimator since these estimators cannot be obtained for some datasets due to a failure of solving $h(\alpha)=0$. 

Finally, another tendency can be seen in the table. The LATE estimate varies less when the ratio $q_{\textit{CO-PR}}/q_{\textit{CO-AS}}$ gets smaller. This is because, as the ratio decreases, the contribution of the \textit{CO-PR} class is smaller. In our simulation scenarios, the ratio is 1 in the first scenario, 0.25 in the second and 0.025 in the third. As can be seen in Table~\ref{tab:sim}, S3 has the narrowest range of the LATE estimates.

\section{Example}
\label{s:example}

We apply our sensitivity analysis methods to estimate the causal effect of delivery at high-level NICUs on the risk of BPD to the Pennsylvania NICU data. The treatment variable $D_i$ means that infant $i$ was delivered at a high-level NICU, the survival variable $S_i$ is an indicator whether infant $i$ survived after his/her hospitalization, and the outcome variable $Y_i$ is the presence of BPD disease for infant $i$. As discussed in Section~\ref{ss:data}, the binary excess travel time indicator $Z_i$ is used as an IV. We found that living close to high-level NICUs is strongly correlated with delivery at high-level NICUs in the premature infant population (infants with a gestational age 23-37 weeks and a birth weight 400-8000g), but is comparatively weakly correlated in the VLBW infant population (infants with a birth weight 500-1500g). More precisely, in the premature infant population, the estimated probability of being delivered at high-level NICUs is 0.77 when $Z=1$, but this probability is 0.38 when $Z=0$. However, in the VLBW infant population, the probability is 0.83 when $Z=1$ and 0.70 when $Z=0$. Since infants in the VLBW infant population were high-risk premature infants, it is more likely that they were delivered at high-level NICUs no matter how far they lived from such NICUs. 

We first present a naive analysis with the Pennsylvania data that does not consider unmeasured confounding. A naive comparison without considering any covariate information (and the IV) indicates that 2.48\% of premature infants who were delivered in high-level NICUs suffer from BPD disease while  0.68\% of premature infants delivered in low-level NICUs suffer from BPD in the premature infant population. Similarly, in the VLBW population, the prevalence of BPD is 20.31\% at high-level NICUs that is much higher than  14.42\% at low-level NICUs. This naive analysis implies that being delivered at high-level NICUs is harmful and make infants have more complications. However, this analysis is misleading because there might be many surviving premature infants delivered at high-level NICUs who would not survive if they were delivered at low-level NICUs. In other words, there might be some infants who barely survived due to being delivered at high-level NICUs, but were not healthy enough to avoid later complications. 
 
 \begin{figure}
	\centering
	\includegraphics[width=160mm]{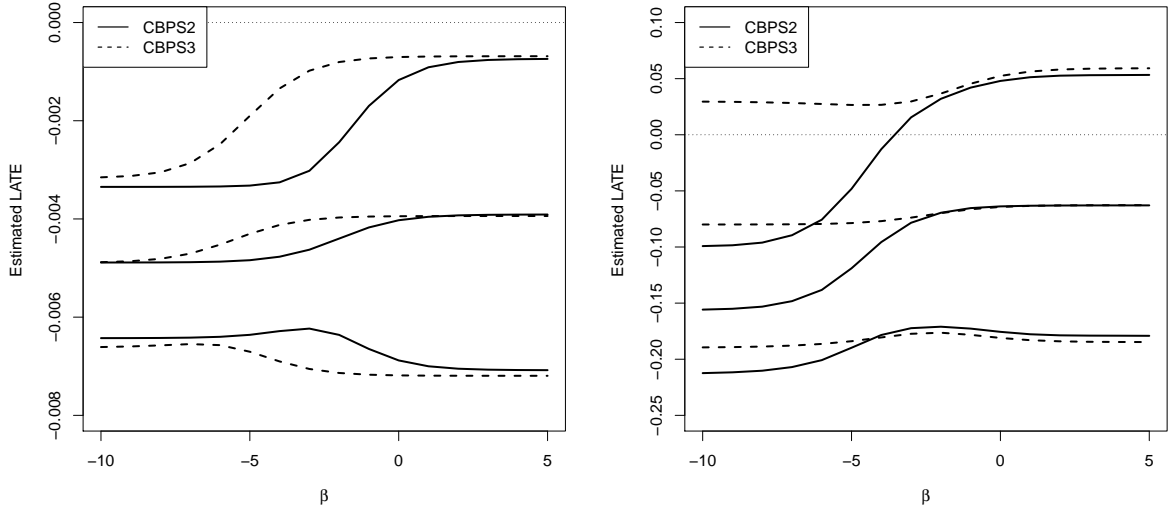}
	\caption{Sensitivity analysis estimates and 95\% confidence intervals of the LATE for the NICU hospital data.}
	\label{fig:late_sensitivity}
\end{figure}

 To make a fair comparison between high-level and low-level NICUs, we consider the combined stratum \textit{CO-AS}, and apply our proposed method to the Pennsylvania NICU data. The proportions of the \textit{CO-AS} and \textit{CO-PR} strata are estimated as 36.2\% and 0.2\% in the premature infant population and 11.6\% and 0.2\% in the VLBW infant population. The ratio between the two strata is 0.005 and 0.017 respectively. These ratios are smaller than the ratio in Simulation 3 in Section~\ref{s:simulation}, so we consider two estimation approaches for the instrumental score $e(\mathbf{X})$, CBPS2 and CBPS3. Figure~\ref{fig:late_sensitivity} shows two sensitivity analysis estimates for the two study populations for $\beta$ in $[-10, 5]$; The left plot is for the premature infant population and the right plot is for the VLBW infant population. The range of $(-\infty, \infty)$ for $\beta$ was initially chosen and examined for the analysis, but the estimated LATE for $\beta$ below -10 (or above 5) flattens out, thus we only plot the cases when $\beta$ is in the interval $[-10, 5]$. In the premature infant population, for all $\beta$ in $[-10, 5]$, the LATE and the 95\% confidence interval are below zero for both CBPS2 and CBPS3 approaches. The CBPS2 and CBPS3 provide slightly different estimate in $-7< \beta <0$, but their estimates converge at the ends of the range. In addition, CBPS2 provides the narrower 95\% confidence intervals. Contrary to the previous naive analysis, both estimation approaches show that there is evidence that delivery at high-level NICUs actually reduces the risk of BPD disease for all values of $\beta$ in the premature infant population. In the VLBW infant population, the right plot of Figure~\ref{fig:late_sensitivity} shows the similar results. The CBPS2 estimation approach provides significant evidence that, for $\beta \leq -4$, delivery at high-level NICUs reduces the risk of BPD, but for $\beta > -4$, the evidence is not significant. This is because we have a smaller sample for the VLBW infant population, and mainly because the proportion of the \textit{CO-AS} is much smaller. If our IV has a stronger effect on treatment, we may draw the same conclusion for all values of $\beta$. Interestingly, the CBPS3 estimation approach provide comparatively larger estimates for $\beta < -4$, and all the LATE estimates have similar values. This is due to the separate estimation of $\alpha$ and $e(\mathbf{X})$. We found that CBPS3 puts too much efforts to solve the equation $h(\alpha)=0$ for lower values of $\beta$, which leads to over-estimation of $\alpha$ as discussed in our simulation studies.   Therefore, the CBPS3 estimation approach is probably inadequate for low values of $\beta$.

\section{Discussion}
\label{s:discussion}

In this article, we have proposed a new method for estimating the average treatment effect when the outcome is censored by death in observational studies. Using IV models, our method conducts sensitivity analysis of the LATE, $E[Y(1,1)-Y(0,1)| C= \textit{CO-AS}]$. To estimate the LATE in a problem with censoring by death problem, we consider an assumption for the relationship between the \textit{CO-AS} and \textit{CO-PR} subpopulations, and add a model restriction to describe this relationship. Building on the CBPS estimation method of Imai and Ratkovic\cite{imai2014}, we incorporate this model restriction within the propensity score estimation procedure so that the parameters can be estimated simultaneously. As shown in our simulation studies, the simultaneous estimation can be the solution when the separate estimation fails to produce an estimate. Also, in the sensitivity analysis, it is shown that the range of the LATE estimate is much narrower when using the simultaneous estimation. We have illustrated our method with the Pennsylvania NICU data with sensitivity analysis. Also, we show that there is evidence that delivery at high-level NICUs reduces the risk of BPD. Although our method is described with binary outcomes, we note that our method can be applied for continuous or discrete outcomes.  

Also, our method can be extended to dealing with multi-valued or even continuous instruments (or treatments). Imai and Ratkovic \cite{imai2014} describes the extension of their CBPS estimation method with multi-valued instruments. In binary instrument cases, covariate balance is optimized by comparing $Z=0$ and $Z=1$ groups. Similarly, in multi-valued instruments, the propensity score can be obtained by optimizing covariate balance with an adjacent instrument value. Fong et al. \cite{fong2018} further generalizes this method for continuous instruments both parametrically and non-parametrically. In future works, we will generalize our method for continuous instruments for the LATE estimation. 
 
Our method uses an inverse probability weighting estimator to identify the LATE given that the sensitivity parameter $\beta$ is fixed and known, but alternatively, one can use outcome regression. However, in our data structure, many model specifications are needed for describing data generating process. For instance, without the instrument, Shepherd et al. \cite{shepherd2006} estimates the LATE using two different versions of modeling the data structure, and both models require at least four model assumptions. The number of the required model assumptions would rapidly increase within IV frameworks. This can make the conclusion based on outcome regression more vulnerable to model misspecifications. However, in our method, the only required model specification is about the instrument propensity score. Also, our simultaneous estimation method enables the robust estimation of the instrument propensity score by balancing covariates between the $Z=1$ and $Z=0$ groups.

\appendix

\section{Proofs of Theorem 1 and Proposition 1}
\label{s:appendix}

\subsection{Proof of Theorem~\ref{thm1}}
\label{ss:proof_thm1}

For Part (i), the IV assumption implies
\begin{align*}
& E[Y(0,1) | S(0)=1, D(0)<D(1), X]\\
&= E[Y(0,1) \cdot S(0) \cdot (D(1)-D(0)) | X] \cdot \frac{1}{\Pr(S(0)=1, D(0)<D(1)| X)} \\
&= \left\{ E[Y(0,1) \cdot S(0) \cdot (1-D(0)) | X] - E[Y(0,1) \cdot S(0) \cdot (1-D(1)) | X] \right\} \cdot \frac{1}{\Pr(S(0)=1, D(0)<D(1)| X)} \\
&= \left\{E[YS(1-D)|X, Z=0] - E[YS(1-D)|X, Z=1] \right\} \cdot \frac{1}{\Pr(S(0)=1, D(0)<D(1)| X)} \\
&= \left\{E\left[\frac{YS(1-D)(1-Z)}{1-e}\bigg| X \right] - E\left[\frac{YS(1-D)Z}{e}\bigg| X \right] \right\}\cdot \frac{1}{\Pr(S(0)=1, D(0)<D(1)| X)} \\
\end{align*}
where $e = \Pr(Z=1|X)$. 

Also, the marginal expectation is 
\begin{align*}
& E[Y(0,1) | S(0)=1, D(0)<D(1)] \\
&= \int E[Y(0,1) | S(0)=1, D(0)<D(1), X=x] \cdot dP(X=x|S(0)=1, D(0)<D(1)) \\
&= \frac{1}{\Pr(S(0)=1, D(0)<D(1))} \times \int \left\{E\left[\frac{YS(1-D)(1-Z)}{1-e}\bigg| X \right] - E\left[\frac{YS(1-D)Z}{e}\bigg| X \right] \right\} dP(X) \\
&= \frac{1}{\Pr(S(0)=1, D(0)<D(1))} \times \left\{ E\left[\frac{YS(1-D)(1-Z)}{1-e} \right] - E\left[\frac{YS(1-D)Z}{e}\right] \right\}.
\end{align*}

Similarly, for Part (ii), the marginal expectation of the potential outcome for the $C_1=CO, C_2=AS \>\text{or}\> PR$ subpopulation is 
\begin{align}
E[g(Y(1,1)) | S(1)=1, D(0)<D(1)] &= \frac{1}{\Pr(S(1)=1, D(0)<D(1))} \times \left\{ E\left[\frac{g(Y)SDZ}{e} \right] - E\left[\frac{g(Y)SD(1-Z)}{1-e}\right] \right\}
\label{eqn:g}
\end{align}
for any function $g(\cdot)$. Also, Assumption~\ref{assumption2} implies 
\begin{align*}
w(y;\beta) &= \frac{f(Y(1,1)=y, S(0)=1, D(0)<D(1))}{f(Y(1,1)=y, S(1)=1, D(0)<D(1))}.
\end{align*}
Therefore, the marginal expectation of the potential outcome for the $C_1=CO, C_2=AS$ subpopulation is
\begin{align*}
&E[Y(1,1) | S(0)=1, D(0)<D(1)] \\
&= \int y \cdot f(Y(1,1)=y \>| S(0)=1, D(0)<D(1)) \> dy \\
&= \int y \cdot\frac{f(Y(1,1)=y, S(0)=1, D(0)<D(1))}{\Pr(S(0)=1, D(0)<D(1))} \> dy\\
&= \int y \cdot\frac{f(Y(1,1)=y, S(1)=1, D(0)<D(1)) \cdot w(y; \beta)}{\Pr(S(0)=1, D(0)<D(1))} \> dy\\
&=  \frac{\Pr(S(1)=1, D(0)<D(1))}{\Pr(S(0)=1, D(0)<D(1))} \int y \cdot w(y; \beta)\cdot  f(Y(1,1)=y \>| S(1)=1, D(0)<D(1))  \> dy \\
&= \frac{\Pr(S(1)=1, D(0)<D(1))}{\Pr(S(0)=1, D(0)<D(1))} \cdot E[Y(1,1) w(Y(1,1); \beta) | S(1)=1, D(0)<D(1)].
\end{align*}
If we choose $g(y) = y w(y;\beta)$ in equation~\eqref{eqn:g}, then 
$$
E[Y(1,1) | S(0)=1, D(0)<D(1)]  = \frac{E\left[\frac{Yw(Y;\beta)SDZ}{e}\right] - E\left[ \frac{Yw(Y;\beta)SD(1-Z)}{1-e}\right]}{\Pr(S(0)=1, D(0)<D(1))}.
$$
Derivation of Part (iii) is similar to derivation of Part (i). 

\subsection{Proof of Proposition~\ref{prop1}}
\label{ss:proof_prop1}

In Proposition~\ref{prop1}, the second equality is the same as Part (iii) in Theorem~\ref{thm1}. We will derive the first equality.
 
We have the following equality: when $Y(1,1)=y$,
\begin{align*}
&\Pr(S(0)=1 | S(1)=1, D(0)<D(1), Y(1,1)=y) \\
&\times f(Y(1,1)=y | S(1)=1, D(0)<D(1)) \\
&\times \Pr(S(1)=1, D(0)<D(1)) \\
&=f(Y(1,1)=y | S(0)=1, D(0)<D(1)) \times \Pr(S(0)=1, D(0)<D(1)).
\end{align*}
After taking integration over $y$, we have
\begin{align}
E(w(Y(1,1); \beta) | S(1)=1, D(0)<D(1)) \times \Pr(S(1)=1, D(0)<D(1)) =\Pr(S(0)=1, D(0)<D(1)).
\label{eqn:equality}
\end{align}
From equation~\eqref{eqn:g}, we have 
$$
E(w(Y(1,1); \beta) | S(1)=1, D(0)<D(1)) = \frac{E\left[\frac{w(Y;\beta)SDZ}{e}\right] - E\left[ \frac{w(Y;\beta)SD(1-Z)}{1-e}\right]}{\Pr(S(0)=1, D(0)<D(1))}
$$
and equation~\eqref{eqn:equality} can be represented by $E\left[\frac{w(Y;\beta)SDZ}{e}\right] - E\left[ \frac{w(Y;\beta)SD(1-Z)}{1-e}\right] = \Pr(S(0)=1, D(0)<D(1)).$

\nocite{*}

\bibliographystyle{plain}
\bibliography{bpd}%

\begin{thebibliography}{10}

\bibitem{abadie2003}
A~Abadie.
\newblock Semiparametric instrumental variable estimation of treatment response
  models.
\newblock {\em J Econom}, 113(2):231--263, 2003.

\bibitem{angrist1996}
JD~Angrist, GW~Imbens, and DB~Rubin.
\newblock Identification of causal effects using instrumental variables.
\newblock {\em J Am Stat Assoc}, 91(434):444--455, 1996.

\bibitem{baiocchi2010}
M~Baiocchi, DS~Small, SA~Lorch, and PR~Rosenbaum.
\newblock Building a stronger instrument in an observational study of perinatal
  care for premature infants.
\newblock {\em J Am Stat Assoc}, 105(492):1285--1296, 2010.

\bibitem{cheng2017}
J~Cheng and W~Lin.
\newblock Understanding causal distributional and subgroup effects with the
  instrumental propensity score.
\newblock {\em Am J Epidemiol}, 2017.

\bibitem{fong2018}
C~Fong, C~Hazlett, and K~Imai.
\newblock Covariate balancing propensity score for a continuous treatment:
  application to the efficacy of political advertisements.
\newblock {\em Ann Appl Stat}, 2018 (forthcoming).

\bibitem{frangakis2002}
CE~Frangakis and DB~Rubin.
\newblock Principal stratification in causal inference.
\newblock {\em Biometrics}, 58(1):21--29, 2002.

\bibitem{gilbert2003}
PB~Gilbert, RJ~Bosch, and MG~Hudgens.
\newblock Sensitivity analysis for the assessment of causal vaccine effects on
  viral load in hiv vaccine trials.
\newblock {\em Biometrics}, 59(3):531--541, 2003.

\bibitem{guo2014}
Z~Guo, J~Cheng, SA~Lorch, and DS~Small.
\newblock Using an instrumental variable to test for unmeasured confounding.
\newblock {\em Stat Med}, 33(20):3528--3546, 2014.

\bibitem{hansen1982}
LP~Hansen.
\newblock Large sample properties of generalized method of moments estimators.
\newblock {\em Econometrica}, 50(4):1029--1054, 1982.

\bibitem{howell2002}
EM~Howell, D~Richardson, P~Ginsburg, and B~Foot.
\newblock Deregionalization of neonatal intensive care in urban areas.
\newblock {\em Am J Public Health}, 92(1):119--124, 2002.

\bibitem{imai2008}
K~Imai.
\newblock Sharp bounds on the causal effects in randomized experiments with
  ``truncation-by-death''.
\newblock {\em Stat Probab Lett}, 78(2):144--149, 2008.

\bibitem{imai2014}
K~Imai and M~Ratkovic.
\newblock Covariate balancing propensity score.
\newblock {\em J R Stat Soc: Ser B}, 76(1):243--263, 2014.

\bibitem{jemiai2007}
Y~Jemiai, A~Rotnitzky, BE~Shepherd, and PB~Gilbert.
\newblock Semiparametric estimation of treatment effects given base-line
  covariates on an outcome measured after a post-randomization event occurs.
\newblock {\em J R Stat Soc: Ser B}, 69(5):879--901, 2007.

\bibitem{jensen2015}
EA~Jensen and SA~Lorch.
\newblock Effects of a birth hospital's neonatal intensive care unit level and
  annual volume of very low-birth-weight infant deliveries on morbidity and
  mortality.
\newblock {\em JAMA Pediatrics}, 169(8):e151906, August 2015.

\bibitem{lapcharoensap2015}
W~Lapcharoensap, SC~Gage, P~Kan, J~Profit, GM~Shaw, JB~Gould, DK~Stevenson,
  H~O'Brodovich, and HC~Lee.
\newblock Hospital variation and risk factors for bronchopulmonary dysplasia in
  a population-based cohort.
\newblock {\em JAMA Pediatrics}, 169(2):e143676, February 2015.

\bibitem{lorch2012}
SA~Lorch, M~Baiocchi, CE~Ahlberg, and DS~Small.
\newblock The differential impact of delivery hospital on the outcomes of
  premature infants.
\newblock {\em Pediatrics}, 130(2):270--278, 2012.

\bibitem{neyman1923}
J~Neyman.
\newblock On the application of probability theory to agricultural experiments.
\newblock {\em Stat Sci}, 5(4):465--480, 1923, 1990.

\bibitem{phibbs2007}
CS~Phibbs, LC~Baker, AB~Caughey, B~Danielsen, SK~Schmitt, and RH~Phibbs.
\newblock Level and volume of neonatal intensive care and mortality in
  very-low-birth-weight infants.
\newblock {\em N Engl J Med}, 356(21):2165--2175, 2007.

\bibitem{phibbs1996}
CS~Phibbs, JM~Bronstein, E~Buxton, and RH~Phibbs.
\newblock The effects of patient volume and level of care at the hospital of
  birth on neonatal mortality.
\newblock {\em JAMA}, 276(13):1054--1059, October 1996.

\bibitem{robins1995}
JM~Robins.
\newblock An analytic method for randomized trials with informative censoring:
  Part 1.
\newblock {\em Lifetime Data Anal}, 1(3):241--254, Sep 1995.

\bibitem{rubin1974}
DB~Rubin.
\newblock Estimating causal effects of treatments in randomized and
  nonrandomized studies.
\newblock {\em J Educ Psychol}, 66(5):688--701, October 1974.

\bibitem{shepherd2006}
BE~Shepherd, PB~Gilbert, Y~Jemiai, and A~Rotnitzky.
\newblock Sensitivity analyses comparing outcomes only existing in a subset
  selected post-randomization, conditional on covariates, with application to
  hiv vaccine trials.
\newblock {\em Biometrics}, 62(2):332--342, 2006.

\bibitem{shepherd2007}
BE~Shepherd, PB~Gilbert, and T~Lumley.
\newblock Sensitivity analyses comparing time-to-event outcomes existing only
  in a subset selected postrandomization.
\newblock {\em J Am Stat Assoc}, 102(478):573--582, 2007.

\bibitem{tan2006}
Z~Tan.
\newblock Regression and weighting methods for causal inference using
  instrumental variables.
\newblock {\em J Am Stat Assoc}, 101(476):1607--1618, 2006.

\bibitem{warner2004}
B~Warner, MJ~Musial, T~Chenier, and E~Donovan.
\newblock The effect of birth hospital type on the outcome of very low birth
  weight infants.
\newblock {\em Pediatrics}, 113(1):35--41, 2004.

\bibitem{yang2016}
F~Yang and DS~Small.
\newblock Using post-outcome measurement information in censoring-by-death
  problems.
\newblock {\em J R Stat Soc: Ser B}, 78(1):299--318, 2016.

\end{thebibliography}

\end{document}